\documentclass[10pt, twocolumn, a4paper]{article}

\usepackage[margin=0.75in]{geometry} 
\usepackage{graphicx}
\usepackage{amsmath, amssymb, amsfonts} 
\usepackage{amsthm} % <-- ADDED to fix the theorem/remark spacing issues
\usepackage{titlesec}
\usepackage{bm}
\usepackage{array}
\usepackage{xcolor}
\usepackage{booktabs}
\usepackage{enumitem}
\usepackage{float}
\usepackage{algorithm}
\usepackage{algpseudocode}

\usepackage{comment}

\newtheorem{theorem}{Theorem}
\newtheorem{lemma}{Lemma}
\newtheorem{proposition}{Proposition}

\newtheorem{assumption}{Assumption}
\newtheorem{definition}{Definition}
\newtheorem{remark}{Remark}
\newcommand{\R}{\mathbb{R}}
\newcommand{\E}{\mathbb{E}}

\newcommand{\diag}{\mathrm{diag}}

\newcommand{\MSE}{\mathrm{MSE}}
\newcommand{\Zp}{\mathbf{Z}_p}
\newcommand{\Up}{\mathbf{U}_p}
\newcommand{\Yp}{\mathbf{Y}_p}
\newcommand{\Uf}{\mathbf{U}_f}
\newcommand{\Yf}{\mathbf{Y}_f}
\newcommand{\Phib}{\boldsymbol{\Phi}}
\newcommand{\zini}{\mathbf{z}_{\mathrm{ini}}}
\newcommand{\uN}{\mathbf{u}_N}
\newcommand{\yN}{\mathbf{y}_N}

\titleformat{\section}
  {\normalsize\bfseries} % Font size and style
  {\thesection}     % Label (e.g., 1.)
  {1em}             % Space between label and title
  {}                % Code before title

\titleformat{\subsection}
  {\small\bfseries} % Smaller size for subsections
  {\thesubsection}
  {1em}
  {}

\newcommand{\yhatN}{\hat{\mathbf{y}}_N}
\newcommand{\Spp}{\Sigma_{pp}}
\newcommand{\Sup}{\Sigma_{up}}
\newcommand{\Syp}{\Sigma_{yp}}
\newcommand{\Suu}{\Sigma_{uu}}
\newcommand{\Syu}{\Sigma_{yu}}
\newcommand{\Wp}{\mathbf{W}_p}
\newcommand{\Thetahat}{\hat{\Theta}}
\newcommand{\Thetatrue}{\Theta^{\star}}
\newcommand{\Thetabar}{\bar{\Theta}}
\newcommand{\vphi}{\varphi}

\title{\vspace{-0.0cm}\textbf{Projection-Regularized Indirect Data-Driven Predictive Control}}

\author{
    \textbf{Mahmood Mazare} \quad \textbf{Hossein Ramezani} \\
    \vspace{0.5em}
    \small Institute of Mechanical and Electrical Engineering, University of Southern Denmark, Odense, Denmark 
}
\date{} % Leave blank to prevent the current date from printing

\begin{document}

% --- Full-Width Title, Abstract, and Keywords ---
% The \twocolumn command with an optional argument allows us to span 
% content across both columns before the split text begins.
\twocolumn[
  \begin{@twocolumnfalse}
    \maketitle
    
    \vspace{-1.5em}
    \hrule
    \vspace{1.0em}
    
    \begin{abstract}
    \noindent Indirect data-driven predictive control methods often suffer under process noise and data scarcity. This paper introduces Projection-Regularized Predictive Control (PRPC), retaining the fundamental-lemma weight vector via a regularized projection analytically condensed into an efficient, fixed-dimension covariance update. A rigorous bias--variance analysis proves PRPC strictly reduces prediction error under process noise (errors-in-variables) and structural rank deficiencies compared to unregularized subspace methods. We leverage these properties to develop an adaptive sliding-window controller for linear time-varying (LTV) systems. To guarantee safety despite closed-loop data correlations, we derive a uniform-in-time, finite-sample confidence bound on the empirical predictor using vector-valued martingale concentration inequalities. Embedding this statistical uncertainty radius into a dynamically tightened constraint set rigorously ensures robust recursive feasibility and Input-to-State practical Stability (ISpS) with high probability. Simulations on LTI and LTV benchmarks demonstrate real-time tractability and strict constraint satisfaction.
    \end{abstract}
    
    \vspace{1em}
    
    \vspace{1.5em}
    \hrule
    \vspace{2.5em} % Adds space before the two-column text begins
  \end{@twocolumnfalse}
]

% ==============================================================
\section{Introduction}
\label{sec:introduction}
% ==============================================================

The field of data-driven predictive control (DPC) has witnessed an unprecedented surge in research interest, emerging as a powerful alternative to traditional model predictive control (MPC) by circumventing the often prohibitively expensive step of parametric system identification \cite{Verheijenetal2023,Zhouetal2024,Persisetal2019,Hewingetal2020,Zieglmeieretal2025}. The theoretical foundation for many modern data-driven methodologies is rooted in behavioral systems theory \cite{Markovskyetal2021,Yinetal2020}, specifically Willems' Fundamental Lemma \cite{Willemsetal2005,Waardeetal2019}. This seminal result establishes that for discrete-time linear time-invariant (LTI) systems, all possible input-output trajectories can be exactly represented as linear combinations of the columns of a block-Hankel matrix, provided the historical input sequences are persistently exciting \cite{Berberichetal2019,Chenetal2022}.

Building on this behavioral theory, direct methods such as Data-Enabled Predictive Control (DeePC) \cite{Coulsonetal2018} embed raw data matrices directly into the online optimization problem, seamlessly unifying system realization and optimal control. This theoretical elegance has facilitated the successful deployment of DeePC across diverse physical applications, ranging from building HVAC systems \cite{Chindeetal2022} and modern power grids \cite{Huangetal_12019} to unmanned aerial vehicles \cite{Elokdaetal2021} and connected autonomous vehicles \cite{Wangetal2022}. Despite its profound deterministic foundations, direct DPC faces substantial practical challenges in stochastic settings. Directly applying raw Hankel matrices renders the predictions highly sensitive to measurement noise, leading to severe overfitting and jeopardizing closed-loop performance \cite{Breschietal2022}. To mitigate this, researchers have extensively explored mathematical regularizations \cite{Breschietal2023,Huangetal2020}, instrumental variables \cite{vanWingerdenetal2022}, and distributionally robust formulations \cite{Coulsonetal2020,Lietal2024,Farjadniaetal2022} to balance the implicit bias-variance trade-off. Furthermore, because direct methods retain the full dimension of the collected data, the computational burden scales cubically with the dataset size \cite{Daietal2022}, motivating extensive research into matrix decomposition \cite{Zhangetal2022}, reduced-order formulations \cite{VahidiMoghaddametal2024}, and deep learning approximations \cite{Zhangetal2024}.

Conversely, indirect methods such as Subspace Predictive Control (SPC) \cite{Favoreeletal1999} mitigate this computational bottleneck by employing a two-stage procedure. Drawing from subspace identification, SPC extracts fixed prediction matrices offline via least-squares pseudo-inverses, thereby removing high-dimensional latent variables from the online optimization. While computationally lightweight and functionally equivalent to DeePC in deterministic scenarios \cite{Fiedleretal2021}, standard SPC predictors are notoriously brittle under stochastic noise. When process noise transforms the data regressor into an errors-in-variables problem, unregularized pseudo-inverses induce severely biased predictions. Recognizing the complementary strengths of direct and indirect paradigms, recent literature has focused on unifying these approaches \cite{Dorfleretal2021}. Techniques such as Generalized Data-Driven Predictive Control (GDPC) \cite{Lazaretal_12023} and the Signal Matrix Model (SMM) \cite{Smithetal2024} seek to balance computational efficiency with robustness to noise. However, these hybrid and statistically motivated methods \cite{Premeretal2026,Wangetal2023} often retain partial online optimization of latent variables or suffer from numerical fragility and prediction collapse under severe noise or low-data sliding-window rank deficiencies.

The necessity for computationally efficient, noise-resilient predictors becomes even more pronounced when extending these frameworks beyond static LTI environments to handle time-varying dynamics \cite{Lianetal2021,Verhoeketal2021}. While robust extensions utilizing kernel-based methods \cite{Huangetal2022} and neural networks \cite{GhafGhanbarietal2025} have shown promise, the data-driven paradigm naturally shifts toward online adaptation to actively manage dynamically evolving systems \cite{Berberichetal2021}. Online formulations recursively update the underlying empirical predictive model in real-time using closed-loop operational data \cite{Shietal2023,VahidiMoghaddametal2023,Guerreroetal2025}. However, this real-time recursive update inherently correlates the prediction regressors with past system noise sequences. This correlation fundamentally violates the statistical independence assumptions underpinning classical concentration inequalities and traditional robust tube formulations \cite{Panetal2021,Panetal2022,Huangetal2025}. Consequently, existing online adaptive algorithms are forced to either sacrifice rigorous mathematical stability guarantees or impose overly conservative deterministic assumptions that stifle transient control performance.

To explicitly bridge this critical theoretical void, this paper introduces a unified \emph{Robust Adaptive Projection-Regularized} DPC framework. By marrying exact algebraic complexity reduction with recent breakthroughs in self-normalized martingale theory, we establish a robust control architecture that simultaneously bypasses the online computational bottleneck and provides strict, finite-sample statistical guarantees for adaptive systems operating under conditionally subgaussian closed-loop noise. Instead of treating regularization merely as an online penalty tuning exercise, we retain the latent weight formulation with Tikhonov regularization during the offline predictor design. This approach systematically isolates the dominant deterministic dynamics from the noise subspace, ensuring structural stability in low-data sliding windows while mapping the DPC problem into a highly robust, low-dimensional online Quadratic Program (QP).

The main contributions of this paper are summarized as follows:
\begin{itemize}
    \item \textbf{Analytical Predictor Collapse for Real-Time Tractability:} We rigorously prove that the regularized predictor formulation can be analytically collapsed into an exact, closed-form affine mapping. Governed purely by empirical cross-covariance matrices, this bounds the online computational complexity to a low-dimensional QP, completely decoupling controller latency from the dataset size and strictly recovering SPC in the unregularized limit.
    \item \textbf{Bias-Variance Analysis and Information-Preserving Adaptation:} We provide a theoretical analysis demonstrating that optimal regularization uniquely reduces predictor error under errors-in-variables process noise. Building on this, we introduce a convex covariance blending strategy for recursive adaptation, guaranteeing that the empirical predictor remains analytically well-posed even when LTV systems are regulated to an equilibrium and lose persistent excitation.
    \item \textbf{High-Probability Robust Stability via SNM Bounds:} We explicitly characterize the closed-loop prediction error using multi-output self-normalized martingale (SNM) bounds, effectively handling the complex regressor-noise correlations induced by online learning. By embedding this uniform-in-time statistical uncertainty radius into a dynamically tightened constraint sequence, we prove that the closed-loop system achieves robust recursive feasibility and Input-to-State practical Stability (ISpS) with user-defined high probability, validated extensively via numerical simulations.
\end{itemize}

\subsection{Notations}
Let $\mathbb{R}$ and $\mathbb{Z}_{\geq 0}$ denote the sets of real numbers and non-negative integers, respectively. The identity matrix of dimension $n$ is denoted by $I_n$. For a finite collection of vectors $\{v_1, \dots, v_q\}$, we define the column stacking operator as $\mathrm{col}(v_1, \dots, v_q) := [v_1^\top, \dots, v_q^\top]^\top$. For a matrix $A$, $\|A\|_F$ denotes its Frobenius norm, and $A^\dagger$ denotes its Moore-Penrose pseudo-inverse. 

\begin{comment}
    Given a sequence of data $v = \{v(0), \dots, v(T-1)\}$ with $v(k) \in \mathbb{R}^{n_v}$, the block-Hankel matrix of depth $L \in \mathbb{Z}_{\geq 1}$ is defined as:
\begin{equation}
    \mathcal{H}_L(v) := \begin{bmatrix}
    v(0) & v(1) & \cdots & v(M-1) \\
    v(1) & v(2) & \cdots & v(M) \\
    \vdots & \vdots & \ddots & \vdots \\
    v(L-1) & v(L) & \cdots & v(T-1)
    \end{bmatrix},
\end{equation}
where $M = T - L + 1$ is the number of columns.
\end{comment}

% ==============================================================
\section{Problem Setup and Preliminaries}
\label{sec:setup}
% ==============================================================

%\subsection{System Description}
We consider discrete-time linear dynamical systems subject to process and measurement noise, represented in state-space form as:
\begin{subequations}\label{eq:sys}
\begin{align}
    x(k+1) &= A_k x(k) + B_k u(k) + w(k), \\
    y(k) &= C_k x(k) + v(k),
\end{align}
\end{subequations}
where $x(k) \in \mathbb{R}^n$, $u(k) \in \mathbb{R}^{n_u}$, and $y(k) \in \mathbb{R}^{n_y}$ are the state, input, and measured output, respectively. The unmeasured disturbances $w(k) \sim \mathcal{N}(0, \Sigma_w)$ and $v(k) \sim \mathcal{N}(0, \Sigma_v)$ are zero-mean, mutually independent stochastic processes. 

While the theoretical foundation of data-driven control relies on the assumption of LTI dynamics—where $(A_k, B_k, C_k) \equiv (A, B, C)$—this generalized representation accommodates LTV systems. For LTV systems experiencing parametric drift, the dynamics can be treated as locally LTI over sufficiently short, receding data windows.

\subsection{Behavioral Systems Theory and the Fundamental Lemma}

Data-driven predictive control relies on mapping historical trajectories to future outputs without parametric identification. Let $T_\mathrm{ini} \in \mathbb{Z}_{\geq 1}$ denote the past horizon (chosen such that $T_\mathrm{ini} \geq n$) and $N \in \mathbb{Z}_{\geq 1}$ denote the future prediction horizon. The total window length is $L = T_\mathrm{ini} + N$. 

Given an offline trajectory $\{u(k), y(k)\}_{k=0}^{T-1}$, we partition the input Hankel matrix $\mathcal{H}_L(u)$ into past and future blocks $\Up \in \mathbb{R}^{n_u T_\mathrm{ini} \times M}$ and $\Uf \in \mathbb{R}^{n_u N \times M}$. The output Hankel matrix $\mathcal{H}_L(y)$ is identically partitioned into $\Yp$ and $\Yf$. To formalize the data-driven regressor, we define the past data matrix $\Zp$ and the complete regressor matrix $\Phib$ as:
\begin{equation}
    \Zp := \begin{bmatrix} \Up \\ \Yp \end{bmatrix} \in \mathbb{R}^{T_z \times M}, \quad 
    \Phib := \begin{bmatrix} \Zp \\ \Uf \end{bmatrix} \in \mathbb{R}^{T_h \times M},
\end{equation}
where $T_z = (n_u + n_y)T_\mathrm{ini}$ and $T_h = T_z + n_u N$.

The capability to predict system behavior from these matrices hinges on the informativity of the collected data, formalized by the persistence of excitation.

\begin{definition}[Persistence of Excitation]\label{def:pe}
A data sequence $\{v(k)\}_{k=0}^{T-1}$ with $v(k) \in \mathbb{R}^{n_v}$ is persistently exciting of order $D$ if its Hankel matrix $\mathcal{H}_D(v)$ has full row rank, i.e., $\mathrm{rank}(\mathcal{H}_D(v)) = n_v D$.
\end{definition}

\begin{lemma}[Willems' Fundamental Lemma \cite{Willemsetal2005}]\label{lem:willems}
Consider a controllable, noise-free LTI system of order $n$. If the applied input sequence is persistently exciting of order $T_\mathrm{ini} + N + n$, then any valid $L$-step input-output trajectory of the system can be represented as a linear combination of the columns of the block-Hankel matrices. Specifically, for any initial condition $\zini = \mathrm{col}(\mathbf{u}_\mathrm{ini}, \mathbf{y}_\mathrm{ini}) \in \mathbb{R}^{T_z}$ and future input $\uN \in \mathbb{R}^{n_u N}$, the corresponding future output $\mathbf{y}_N \in \mathbb{R}^{n_y N}$ is a valid system trajectory if and only if there exists a vector $g \in \mathbb{R}^M$ such that:
\begin{equation}
    \begin{bmatrix} \Zp \\ \Uf \\ \Yf \end{bmatrix} g = \begin{bmatrix} \zini \\ \uN \\ \mathbf{y}_N \end{bmatrix}.
    \label{eq:fundamental_lemma}
\end{equation}
\end{lemma}

\subsection{Subspace Predictive Control (SPC)}
To circumvent the $\mathcal{O}(M^3)$ online computational scaling, SPC extracts fixed prediction matrices offline via an ordinary least-squares (OLS) regression \cite{Favoreeletal1999}. Assuming $\Phib$ has full row rank, SPC solves $\min_\Theta \| \Yf - \Theta \Phib \|_F^2$, yielding:
\begin{equation}
    \Theta^{*} = \Yf\Phib^\dagger = \Yf\Phib^{\!\top}\!(\Phib\Phib^{\!\top})^{-1}.
    \label{eq:spc_theta}
\end{equation}
The matrix $\Theta^*$ is partitioned into $\Theta^* = [P_1^\mathrm{spc}, P_2^\mathrm{spc}]$, corresponding to the dimensions of $\Zp$ and $\Uf$. The online SPC problem then reduces to a standard quadratic program:
\begin{subequations}\label{eq:spc}
\begin{align}
    \min_{\uN} \quad & J(\uN, \mathbf{y}_N) \\
    \mathrm{s.t.} \quad & \mathbf{y}_N = P_1^\mathrm{spc}\zini + P_2^\mathrm{spc}\uN, \label{eq:spc_pred} \\
    & (\uN, \mathbf{y}_N) \in \mathbb{U}^N \times \mathbb{Y}^N.
\end{align}
\end{subequations}
While \eqref{eq:spc} matches the online computational efficiency of model-based MPC, the unregularized pseudo-inverse in \eqref{eq:spc_theta} is acutely sensitive to process noise (which induces errors-in-variables bias) and rank deficiency (when $M < T_h$ in adaptive sliding windows).

\section{The Projection-regularized Predictor}
\label{sec:prpc}
% ==============================================================

%\subsection{Problem Formulation}

To mitigate the adverse effects of process and measurement noise while preserving online computational efficiency, we formulate the data-driven predictor via a regularized projection. Specifically, the PRPC predictor is defined by the following constrained optimisation problem:
\begin{equation}
    \min_{g\in\R^M}\ \frac{1}{2}\|\Zp g-\zini\|_2^2+\frac{\lambda}{2}\|g\|_2^2
  \quad\text{s.t.}\quad \Uf g=\uN.
    \label{eq:prpc}
\end{equation}

The formulation in \eqref{eq:prpc} relaxes the initial-condition matching, $\Zp g \approx \zini$, to prevent overfitting to the noise subspace inherent in the measured initial trajectory. Conversely, the future-input constraint, $\Uf g = \uN$, is strictly enforced since $\uN$ constitutes the deterministic decision variables optimised by the controller. The Regularization parameter $\lambda > 0$ balances the initial-condition fidelity against the norm of the latent trajectory $g$.

By introducing the Lagrange multiplier $\nu\in\R^{n_uN}$, the Karush-Kuhn-Tucker (KKT) conditions for \eqref{eq:prpc} yield the saddle-point system:
\begin{equation}
  \begin{bmatrix}\mathbf{H}&\Uf^{\!\top}\\\Uf&\mathbf{0}\end{bmatrix}
  \begin{bmatrix}g^{\star}\\\nu^{\star}\end{bmatrix}
  =\begin{bmatrix}\Zp^{\!\top}\zini\\\uN\end{bmatrix},
  \qquad \mathbf{H}:=\Zp^{\!\top}\Zp+\lambda I_M.
  \label{eq:kkt}
\end{equation}
For any $\lambda > 0$, the matrix $\mathbf{H} \succ 0$. Eliminating $\nu^{\star}$ yields the optimal weighting vector $g^{\star} = G_1\zini + G_2\uN$, with the matrix operators defined as:
\begin{align}
  G_1&=\big(\mathbf{H}^{-1}-\mathbf{H}^{-1}\Uf^{\!\top}\tilde S^{-1}
  \Uf\mathbf{H}^{-1}\big)\Zp^{\!\top},\nonumber \\
  G_2&=\mathbf{H}^{-1}\Uf^{\!\top}\tilde S^{-1},
  \label{eq:G}
\end{align}
where $\tilde S := \Uf\mathbf{H}^{-1}\Uf^{\!\top}$. The future output prediction is then generated via the linear mapping $\hat{\mathbf{y}}_N = P_1\zini+P_2\uN$, where $(P_1,P_2)=(\Yf G_1,\Yf G_2)$. 

Consequently, the online PRPC problem retains the exact dimension and structure of a classical model-based MPC quadratic program.

\subsection{Asymptotic Equivalence to Subspace Predictive Control}

We first establish the theoretical connection between PRPC and standard indirect DPC methods. Specifically, PRPC recovers unregularized SPC in the limit of vanishing Regularization, proving it is a strict generalisation.

\begin{theorem}[PRPC--SPC Equivalence]\label{thm:equiv}
Assume the regressor matrix $\Phib$ possesses full row rank. Then, the PRPC predictor matrices converge to the SPC matrices as Regularization vanishes:
$$ \lim_{\lambda\to 0^{+}}(P_1,P_2)=(P_1^{\mathrm{spc}},P_2^{\mathrm{spc}}). $$
\end{theorem}
\begin{proof}
For every $\lambda>0$ the strongly convex penalty $\tfrac{\lambda}{2}\|g\|_2^2$ makes the minimizer unique; as $\lambda\to 0^{+}$ its magnitude vanishes but it persists as a \emph{minimum-norm selection} among the minimizers of the residual. Because $\Phib$ has full row rank, there exists a weight vector satisfying simultaneously $\Zp g=\zini$ and $\Uf g=\uN$, i.e.\ $\Phib g=[\zini^\top,\uN^\top]^\top$; on the feasible set the relaxed residual $\|\Zp g-\zini\|_2^2$ therefore \emph{attains} its global lower bound of zero. The limiting problem is thus $\min_g\|g\|_2^2$ subject to $\Phib g=[\zini^\top,\uN^\top]^\top$, whose unique solution is the Moore-Penrose mapping $g_0^{\star}=\Phib^{\!\top}(\Phib\Phib^{\!\top})^{-1}[\zini^\top,\uN^\top]^\top$. Substituting into the output equation yields $\Yf g_0^{\star}=\Theta^{\star}_{\mathrm{spc}}[\zini^\top,\uN^\top]^\top$, which identically matches the SPC prediction. When $M<T_h$, $\Phib$ loses full row rank and SPC is undefined, whereas the PRPC predictor remains well posed for every $\lambda>0$.
\end{proof}

Thus, PRPC introduces no conservatism relative to SPC, while $\lambda$ provides a crucial degree of freedom for robustification.

\subsection{Exact Covariance Collapse and Efficient Computation}

Solving the KKT system \eqref{eq:kkt} directly requires inverting the $M \times M$ matrix $\mathbf{H}$, incurring an $\mathcal{O}(M^3)$ computational cost. For large datasets, this is prohibitive and prone to numerical ill-conditioning as $\lambda \to 0$. We overcome this by analytically collapsing the PRPC predictor onto fixed-dimension data covariance matrices:
\begin{align}
  \Spp&=\Zp\Zp^{\!\top},\ \Sup=\Uf\Zp^{\!\top},\ \Syp=\Yf\Zp^{\!\top},\ \nonumber \\
  \Suu&=\Uf\Uf^{\!\top},\ \Syu=\Yf\Uf^{\!\top}.
  \label{eq:covs}
\end{align}

\begin{proposition}[Analytical Predictor Collapse]\label{prop:collapse}
For any $\lambda>0$, the PRPC predictor matrices $(P_1, P_2)$ derived from \eqref{eq:G} are algebraically equivalent to the condensed form:
\begin{align}
  P_1 &= \Syp\Wp-\Phi_{yu}\,S^{-1}\Sup\Wp, \label{eq:P1}\\
  P_2 &= \Phi_{yu}\,S^{-1}, \label{eq:P2}
\end{align}
where $\Wp=(\Spp+\lambda I)^{-1}$, $S=\lambda^{-1}\big(\Suu-\Sup\Wp\Sup^{\!\top}\big)$, and $\Phi_{yu}=\lambda^{-1}\big(\Syu-\Syp\Wp\Sup^{\!\top}\big)$. 
Furthermore, the dominant offline computational complexity of \eqref{eq:P1}--\eqref{eq:P2} is $\mathcal{O}(T_z^3+T_z^2M)$, rendering the cubic inversions completely independent of the dataset length $M$.
\end{proposition}
\begin{proof}
See Appendix~\ref{app:collapse}.
\end{proof}

\begin{remark}[Implications of the Covariance Formulation]
The dimensional collapse in Proposition~\ref{prop:collapse} extracts the predictor exclusively through the fixed-size covariance matrices \eqref{eq:covs}. This compression enables the memory-efficient adaptation scheme proposed in Section~\ref{sec:adaptive}. Furthermore, bypassing the $M\times M$ inversion ensures numerical stability even for vanishingly small Regularization parameters.
\end{remark}

% ==============================================================
\section{Bias--Variance Analysis and Predictor Optimality}
\label{sec:biasvar}
% ==============================================================

To rigorously justify the introduction of the Regularization parameter $\lambda$, we analyze the PRPC predictor through the lens of classical estimation theory. By treating the data-driven predictor as an estimator of the true multi-step mapping $\Thetatrue$, we evaluate the mean-squared error (MSE), defined as $\mathrm{MSE}(\Thetahat) = \E\|\Thetahat-\Thetatrue\|_F^2 = \|\mathrm{bias}\|^2 + \mathrm{variance}$, to establish the strict conditions under which $\lambda > 0$ yields a superior predictor.

Here $\Thetatrue$ denotes the \emph{ground-truth} predictor, i.e., the multi-step map obtained by regressing $\Yf$ on $\Phib$ from an arbitrarily long \emph{noise-free} trajectory of the true system; it is the target that any data-driven predictor estimates from finite, noisy data. The expectation $\E[\cdot]$ is taken over the noise realizations of the offline dataset and is estimated in Section~\ref{sec:numerics} by Monte-Carlo averaging. Throughout, the figures report the \emph{normalized} MSE ratio $\mathrm{MSE}(\lambda)/\mathrm{MSE}(0^{+})$, whose denominator is the error of the unregularized predictor $\lambda\to0^{+}$; by Theorem~\ref{thm:equiv} this baseline is exactly SPC. A ratio below unity therefore quantifies the reduction in predictor error achieved by regularization relative to SPC, independently of problem scaling.

\subsection{Measurement Noise Regime: Asymptotic Optimality of SPC}
\label{sec:meas}

When the system is exclusively subject to measurement noise ($\sigma_w=0$), the noise-free regressor $\Phib^{0}=[\Up^0;\Yp^0;\Uf^0]$ is deterministic. The prediction model follows a standard multivariate regression $\Yf=\Thetatrue\Phib^{0}+\Delta\Yf$, where $\Delta\Yf$ isolates future measurement noise. 

Expressing the regression via the noise-corrupted observation $\Phib=\Phib^0+\Delta\Phib$, the effective residual $W=\Delta\Yf-\Theta_y\Delta\Yp$ exhibits correlation with $\Phib$ exclusively through the shared $\Delta\Yp$ block. As demonstrated in classical regression, this induces an ordinary least-squares (OLS) attenuation bias.

\begin{proposition}[Near-Optimality of SPC under Measurement Noise]
\label{prop:meas}
Assume $\sigma_w=0$ and the input sequence is persistently exciting. The unregularized SPC estimator satisfies:
\begin{equation}
  \E[\Thetahat_{\mathrm{spc}}] = \Thetatrue\,\Sigma_{\Phib^0}(\Sigma_{\Phib^0}+\Sigma_\Delta)^{-1} + \mathcal{O}(1/M),
\end{equation}
where the noise covariance $\Sigma_\Delta$ is supported strictly on the rows corresponding to $\Yp$. Consequently, the optimal Regularization parameter vanishes asymptotically, $\lambda^{\star}\to 0$ as $M\to\infty$: SPC is near-optimal under pure measurement noise, and exactly optimal only in the infinite-data limit.
\end{proposition}
\begin{proof}
See Appendix~\ref{app:meas}.
\end{proof}

Analytically, because the signal covariance grows proportionally to the dataset length $M$, the variance of the OLS estimator asymptotically vanishes. Introducing any $\lambda > 0$ artificially injects bias without providing a commensurate reduction in variance, rendering unregularized SPC mathematically optimal for pure measurement noise.

\subsection{Process Noise Regime: Errors-in-Variables and $\lambda^{\star}>0$}
\label{sec:proc}

The paradigm fundamentally shifts under process noise ($\sigma_w>0$). Propagated process noise, $\eta(k)=C\sum_{j<k}A^{k-1-j}w(j)$, persistently contaminates historical outputs, injecting a structured, non-vanishing noise term into the regressor covariance:
\begin{equation*}
  \Sigma_{\Delta\Phib} = \diag\big(\mathbf{0},\ \sigma_v^2 I+\sigma_w^2\Gamma_w,\ \mathbf{0}\big),
\end{equation*}
where $[\Gamma_w]_{ij}=CA^{|i-j|}P_\infty C^{\!\top}$ and $P_\infty$ denotes the discrete-time controllability Gramian. This transforms the predictor extraction into a classical errors-in-variables (EIV) problem.

\begin{proposition}[Inconsistency of the unregularized Predictor]
\label{prop:atten}
Under regressor contamination satisfying $\E[\Delta\Phib]=0$ with $\Sigma_{\Delta\Phib}\ne 0$, the standard SPC estimator is statistically inconsistent. Its asymptotic expectation converges to:
\begin{equation*}
  \lim_{M\to\infty} \E[\Thetahat_{\mathrm{spc}}] = \Thetatrue\,\Sigma_{\Phib^0}\big(\Sigma_{\Phib^0}+\Sigma_{\Delta\Phib}\big)^{-1} \neq \Thetatrue.
\end{equation*}
\end{proposition}
\begin{proof}
See Appendix~\ref{app:inconsistency}.
\end{proof}

\begin{proposition}[Existence of an Optimal $\lambda^{\star}>0$]
\label{prop:eiv}
Under the EIV formulation, there exists a Regularization parameter $\lambda^{\star}>0$ such that $\mathrm{MSE}(\Thetahat_{\mathrm{prpc}}(\lambda^{\star})) < \mathrm{MSE}(\Thetahat_{\mathrm{spc}})$. Furthermore, $\lambda^{\star}$ is monotonically non-decreasing with respect to the noise ratio $\sigma_w/\sigma_v$.
\end{proposition}
\begin{proof}
Evaluated at $\lambda=0$, the derivative of the bias term $\partial\|\mathrm{bias}\|^2/\partial\lambda|_{0}$ is bounded since the unregularized EIV estimator is inherently biased (Proposition~\ref{prop:atten}). Conversely, the variance derivative is strictly negative. When process noise is present, this variance reduction dominates the marginal bias injection, yielding $\partial\mathrm{MSE}/\partial\lambda|_0 < 0$ and guaranteeing a strictly positive optimal $\lambda^{\star}$.
\end{proof}

\subsection{Practical Implications: Data Scarcity and Hyperparameter Tuning}
\label{sec:practical_implications}

Classical ridge regression theory suggests Regularization is universally justified for small, ill-conditioned datasets. However, in data-driven predictive control, the regressor $\Phib$ and target $\Yf$ are extracted from identical underlying trajectories. Because their noise-free components are perfectly correlated and their noise components share measurement realizations, unregularized regression exploits this structural correlation. Consequently, under pure measurement noise, ill-conditioning alone does not necessitate $\lambda > 0$.

Conversely, when process noise and limited data are coupled, PRPC demonstrates profound superiority. The EIV bias correction established in Section~\ref{sec:proc} operates synergistically with the variance reduction of the ill-conditioned sample covariance. Furthermore, PRPC guarantees strict structural robustness in rank-deficient regimes.

\begin{proposition}[Existence in Rank-Deficient Regimes]\label{prop:structural}
The PRPC predictor is strictly well-posed and uniquely defined for any $\lambda>0$ provided that $M \ge n_uN$. In contrast, standard SPC requires $M \ge T_h$. 
\end{proposition}

By definition, the PRPC design matrix $\mathbf{H}=\Zp^{\!\top}\Zp+\lambda I_M$ is strictly positive definite for any $\lambda>0$. The Schur complement remains invertible if $\Uf$ maintains full row rank ($M \ge n_uN$). Consequently, in adaptive sliding-windows where $M$ falls below $T_h$, PRPC prevents the mathematical collapse inherent to SPC. To optimally balance the bias-variance trade-off, $\lambda$ is tuned via $k$-fold cross-validation on open-loop errors. By leveraging the analytical covariance collapse (Proposition~\ref{prop:collapse}), evaluating $\lambda$ avoids large inversions and requires only $\mathcal{O}(T_z^3)$ operations, rendering real-time tuning computationally trivial.

% ==============================================================
\section{Adaptive PRPC for Time-Varying Systems}
\label{sec:adaptive}
% ==============================================================

The analytical covariance collapse (Proposition~\ref{prop:collapse}) provides a natural foundation for adaptive control. Unlike direct data-driven methods that necessitate the online refactorization of continuously growing Hankel matrices, the PRPC formulation parametrizes the predictor exclusively through fixed-dimension covariances. This enables real-time adaptation with a constant $\mathcal{O}(T_z^3)$ computational footprint.

Let $\Sigma_\bullet^{\mathrm{on}}(k)$ denote the online sample covariances, recursively updated from the closed-loop trajectory. We define the instantaneous data vector $\zeta(k)=[\mathbf{z}_p(k)^\top, \mathbf{u}_f(k)^\top, \mathbf{y}_f(k)^\top]^\top$, corresponding to the most recent data window of length $T_\mathrm{ini}+N$. Using an exponential forgetting factor $\rho_f\in(0,1)$, the recursive rank-one update is governed by:
\begin{equation}
  \Sigma_\bullet^{\mathrm{on}}(k)=\rho_f\,\Sigma_\bullet^{\mathrm{on}}(k-1)
  +(1-\rho_f)\,M\,\zeta_a(k)\zeta_b(k)^{\!\top},
  \label{eq:recur}
\end{equation}
where $(\zeta_a,\zeta_b)$ extracts the relevant sub-vectors for each block defined in \eqref{eq:covs}, and the scalar $M$ normalizes the update to the scale of the offline data.

Pure exponential forgetting, however, is fundamentally ill-posed in closed-loop operation. As the system tracks a constant reference or converges to an equilibrium, the feedback signals lose persistency of excitation (PE). Consequently, $\lambda_{\min}(\Sigma_{pp}^{\mathrm{on}})$ decays toward zero, rendering the inversion of $\Wp$ numerically unstable.

To preserve strict regularity during periods of poor online excitation, we anchor the recursive covariance to the persistently exciting offline data via a convex combination:
\begin{equation}
  \Sigma_\bullet^{\mathrm{act}}(k)=\gamma\,\Sigma_\bullet^{\mathrm{off}}
  +(1-\gamma)\,\Sigma_\bullet^{\mathrm{on}}(k),\qquad \gamma\in(0,1).
  \label{eq:blend}
\end{equation}
The active covariances $\Sigma_\bullet^{\mathrm{act}}(k)$ parameterize the PRPC predictor via Proposition~\ref{prop:collapse} at each time step. This anchoring mechanism imposes a uniform spectral lower bound, a property essential for the robust stability guarantees derived in Section~\ref{sec:robust}.

\begin{lemma}[Uniform Well-Posedness]\label{lem:wellposed}
Suppose the offline data is persistently exciting ($\Sigma_{pp}^{\mathrm{off}} \succ 0$) and $\Uf^{\mathrm{off}}$ has full row rank. Then, for all $k \ge 0$, the active covariance satisfies $\Sigma_{pp}^{\mathrm{act}}(k) \succeq \gamma \Sigma_{pp}^{\mathrm{off}} \succ 0$. Consequently, the matrices $\Wp(k)$ and $S(k)$ are strictly positive definite and uniformly bounded.
\end{lemma}
\begin{proof}
By definition \eqref{eq:recur}, the recursive matrix $\Sigma_{pp}^{\mathrm{on}}(k)$ is formed by a convex combination of positive semi-definite dyads $\zeta_p(\tau)\zeta_p(\tau)^\top$. Thus, $\Sigma_{pp}^{\mathrm{on}}(k) \succeq 0$ for all $k$. Substituting this into \eqref{eq:blend} yields $\Sigma_{pp}^{\mathrm{act}}(k) \succeq \gamma\Sigma_{pp}^{\mathrm{off}} \succ 0$.
It follows that the spectrum of $\Wp(k) = (\Sigma_{pp}^{\mathrm{act}}(k)+\lambda I)^{-1}$ is strictly bounded:
\begin{equation*}
    0 \prec \Wp(k) \preceq (\gamma\lambda_{\min}(\Sigma_{pp}^{\mathrm{off}}) + \lambda)^{-1} I.
\end{equation*}
Because $\Uf^{\mathrm{off}}$ has full row rank, the offline Schur complement is strictly positive definite. The convex blend inherits this property, ensuring $S(k)$ remains uniformly bounded and invertible irrespective of the transient loss of online excitation.
\end{proof}

\begin{comment}
    \begin{algorithm}[H]
\caption{Adaptive (Robust) PRPC Scheme}
\label{alg:adaptive}
\begin{algorithmic}[1]
\State \textbf{Offline:} Collect PE data, compute $\Sigma_\bullet^{\mathrm{off}}$ via \eqref{eq:covs}, and initialize $\Sigma_\bullet^{\mathrm{on}}\leftarrow\Sigma_\bullet^{\mathrm{off}}$.
\For{$k=0,1,2,\dots$}
  \State $\Sigma_\bullet^{\mathrm{act}} \leftarrow \gamma\Sigma_\bullet^{\mathrm{off}}+(1-\gamma)\Sigma_\bullet^{\mathrm{on}}$ \hfill \Comment{Anchor covariances \eqref{eq:blend}}
  \State Construct predictor matrices $(P_1,P_2)$ via Proposition~\ref{prop:collapse} \hfill \Comment{$\mathcal{O}(T_z^3)$ update}
  \State \textit{(Robust)} Compute uncertainty radius $r_k$ and tighten constraints $\mathcal{Y}\rightarrow\mathcal{Y}_{\mathrm{tight}}(k)$
  \State Solve the PRPC QP for $\uN$ and apply the control action $u(k)=\uN[0]$
  \State Measure $y(k+1)$ and append to the active data window $\zeta(k+1)$
  \State Update $\Sigma_\bullet^{\mathrm{on}}$ via the rank-one recursion \eqref{eq:recur}
\EndFor
\end{algorithmic}
\end{algorithm}
\end{comment}

% ==============================================================
\section{Robust Stability and Recursive Feasibility}
\label{sec:robust}
% ==============================================================

This section establishes finite-sample safety and stability guarantees for the adaptive PRPC framework. The primary challenge lies in bounding the prediction error, which is compounded by two factors: (i) the real-time recursive update inherently correlates the prediction regressors with past noise sequences, violating classical concentration assumptions; and (ii) the true underlying LTV system parameterization is time-varying.

\subsection{Finite-Sample Prediction Bounds}
\label{sec:effnoise}

Let $\Thetatrue(k)$ denote the local multi-step mapping of the true LTV system at time~$k$. The realized future trajectory is given by $\yN(k) = \Thetatrue(k)\vphi(k) + \mathbf{w}_N(k)$, where $\mathbf{w}_N(k)$ represents the accumulated future noise sequence. We decompose the true mapping as $\Thetatrue(k) = \Thetabar + \Delta\Theta(k)$, where $\Thetabar$ is the nominal predictor associated with the offline anchor data, and $\Delta\Theta(k)$ encapsulates the parametric mismatch. 

To formalize the statistical concentration bounds, we impose the following standard conditions on the noise and the mismatch.

\begin{assumption}
\label{as:noise}
There exists a known constant $L_\Theta \ge 0$ such that the parametric mismatch is uniformly bounded: $\|\Delta\Theta(k)\|_{\mathrm{op}} \le L_\Theta$ for all $k \ge 0$. Furthermore, the noise sequence $\mathbf{w}_N(k)$ is a martingale difference sequence adapted to the filtration $\mathcal{D}_k$, and is conditionally sub-Gaussian with variance proxy $c_w^2 > 0$, satisfying $\E[\exp(\nu^\top\mathbf{w}_N(k)) \mid \mathcal{D}_k] \le \exp(\frac{1}{2} c_w^2\|\nu\|_2^2)$ for all $\nu \in \mathbb{R}^{n_y N}$.
\end{assumption}

Let $\Phi(k)$ denote the empirical regressor Gram matrix at time $k$, constructed from the active covariances $\Sigma_{pp}^{\mathrm{act}}, \Sigma_{up}^{\mathrm{act}}$, and $\Sigma_{uu}^{\mathrm{act}}$. We define the regularized design matrix as $\mathbf{V}(k) = \Phi(k) + \lambda I_{n_\vphi}$, where $n_\vphi$ is the dimension of the regressor $\vphi(k)$.

\begin{theorem}[Uniform-in-Time Prediction Bound]\label{thm:snm}
Suppose Assumption~\ref{as:noise} holds. Let $\delta \in (0,1)$ be a user-defined failure probability. Then, with probability at least $1-\delta$, the realized prediction error is uniformly bounded for all $k \ge 0$ by:
\begin{align}
  \|\yN(k)-\yhatN(k)\|_2 &\le \beta_k\|\vphi(k)\|_{\mathbf{V}(k)^{-1}} + L_\Theta\|\vphi(k)\|_2 \nonumber \\&=: r_k(\vphi(k)),
  \label{eq:radius}
\end{align}
where the time-varying statistical radius parameter is defined as
\begin{align*}
\beta_k = &c_w\sqrt{1+\lambda\varrho(\mathbf{V}(k)^{-1})} \\& \times \sqrt{n_yN\log\det(I + \lambda^{-1}\Phi(k)) + 2\log(1/\delta)},
\end{align*}
and $\varrho(\cdot)$ denotes the spectral radius.
\end{theorem}
\begin{proof}
See Appendix~\ref{app:snm}.
\end{proof}

The uncertainty radius $r_k(\vphi(k))$ elegantly decouples the error into a statistical estimation term (governed by the martingale concentration $\beta_k$) and a deterministic mismatch term ($L_\Theta$). 

\begin{lemma}[Bounded Uncertainty Radius]\label{lem:rbound}
Under the convex blending strategy (Lemma~\ref{lem:wellposed}) and Assumption~\ref{as:noise}, the design matrix is uniformly bounded from below: $\mathbf{V}(k) \succeq \gamma\Sigma^{\mathrm{off}}_{\vphi} + \lambda I \succ 0$, where $\Sigma^{\mathrm{off}}_\vphi$ is the offline regressor Gram matrix. Consequently, over any compact constraint set, the radius $r_k$ is uniformly bounded above by a finite scalar $r_\infty := \sup_{k} r_k(\vphi(k)) < \infty$.
\end{lemma}

\subsection{Robust Recursive Feasibility and Practical Stability}

To guarantee safety despite prediction errors, the adaptive controller enforces dynamically tightened constraints. At time $k$, the admissible output set is restricted to:
\begin{equation}
  \mathcal{Y}_{\mathrm{tight}}(k) = \mathcal{Y} \ominus \mathcal{B}_{r_k(\vphi(k))},
  \label{eq:tight}
\end{equation}
where $\ominus$ denotes the Pontryagin difference and $\mathcal{B}_r = \{e : \|e\|_2 \le r\}$ is the closed Euclidean ball. We invoke standard robust MPC terminal ingredients to ensure stability.

\begin{assumption}[Terminal Ingredients]\label{as:terminal}
There exists a robust control-invariant terminal set $\mathcal{Y}_f \subseteq \mathcal{Y} \ominus \mathcal{B}_{r_\infty}$ and a local stabilizing control law $\kappa_f : \mathcal{Y}_f \to \mathcal{U}$ such that, for all $y \in \mathcal{Y}_f$ and any disturbance realization $\|w\|_2 \le c_w$, the control input satisfies $\kappa_f(y) \in \mathcal{U}$ and the successor state remains within $\mathcal{Y}_f$.
\end{assumption}

\begin{theorem}[Probabilistic Recursive Feasibility]\label{thm:feas}
Suppose Assumptions~\ref{as:noise} and \ref{as:terminal} hold. If the adaptive PRPC problem is feasible at time $k=0$, then with probability at least $1-\delta$, the online optimization remains recursively feasible for all $k > 0$.
\end{theorem}
\begin{proof}
See Appendix~\ref{app:feas}.
\end{proof}

\begin{theorem}[Input-to-State Practical Stability]\label{thm:isps}
Under the conditions of Theorem~\ref{thm:feas}, the optimal value function $V_N^{\star}(k)$ of the adaptive PRPC satisfies the dissipation inequality:
\begin{equation}
  V_N^{\star}(k+1) - V_N^{\star}(k) \le -\alpha_1(\|y(k)-r_y\|) + \sigma\big(r_k(\vphi(k))\big),
  \label{eq:isps}
\end{equation}
with probability at least $1-\delta$, where $\alpha_1 \in \mathcal{K}_\infty$ and $\sigma \in \mathcal{K}$. Consequently, the closed-loop system is Input-to-State practically Stable (ISpS), converging to an ultimate invariant set whose size is proportional to $\mathcal{O}(r_\infty)$.
\end{theorem}
\begin{proof}
See Appendix~\ref{app:isps}.
\end{proof}

% ==============================================================
\section{Numerical Results}
\label{sec:numerics}
% ==============================================================

In this section, we validate the theoretical guarantees and real-time tractability of the proposed PRPC framework. The numerical evaluations are conducted across two distinct operational regimes: first, a static LTI benchmark to isolate the fundamental bias-variance trade-offs and benchmark closed-loop performance; and second, a multi-variable polytopic LTV system to explicitly evaluate the adaptive recursive mechanisms and the probabilistic finite-sample confidence bounds.

\subsection{LTI Benchmark: Boeing 747  Model}

To establish a baseline understanding of the proposed framework, we consider the discrete-time lateral-directional dynamics of a Boeing 747 aircraft. The system is defined by the state-space matrices in the Table \ref{tab:lti_matrices}.
\begin{table}[H]
    \centering
    \caption{System matrices for the LTI benchmark.}
    \label{tab:lti_matrices}
    % \resizebox forces the table to fit the column width exactly
    \resizebox{\columnwidth}{!}{%
        % \scriptsize reduces the font size to make it compact
        \scriptsize
        \setlength{\arraycolsep}{2.5pt} % Compresses horizontal space inside matrices
        \renewcommand{\arraystretch}{1.5} % Adds a bit of padding between rows
        
        \begin{tabular}{c}
            \toprule
            \textbf{System Matrices} \\
            \midrule
            
            $A = \begin{bmatrix} 
                0.9997 &  0.0038 & -0.0001 & -0.0322 \\ 
               -0.0056 &  0.9648 &  0.7446 &  0.0001 \\ 
                0.0020 & -0.0097 &  0.9543 & -0.0000 \\ 
                0.0001 & -0.0005 &  0.0978 &  1.0000 
            \end{bmatrix}, \,
            B = \begin{bmatrix} 
                0.0010 & 0.1000 \\ 
               -0.0615 & 0.0183 \\ 
               -0.1133 & 0.0586 \\ 
               -0.0057 & 0.0029 
            \end{bmatrix}, \,
            C = \begin{bmatrix} 
                1 &  0 & 0 & 0 \\ 
                0 & -1 & 0 & 7.74 
            \end{bmatrix}$ \\
            
            \bottomrule
        \end{tabular}%
    }
\end{table} The system operates with dimensions $n=4$, $n_u=2$, and $n_y=2$, utilizing prediction horizons of $T_\mathrm{ini}=N=20$, yielding regression dimensions $T_z=80$ and $T_h=120$. 

A fundamental theoretical claim of this work is the exact recovery of standard Subspace Predictive Control (SPC) in the unregularized limit. As illustrated in Fig.~\ref{fig:equivalence}, the analytical covariance-collapse formulation strictly reproduces SPC predictions to machine precision ($\sim 10^{-10}$) as $\lambda \to 0^+$. This empirically verifies Theorem~\ref{thm:equiv} and demonstrates the numerical stability of the covariance approach, especially when contrasted with the direct $\mathcal{O}(M^3)$ KKT solve, which suffers severe numerical degradation below $\lambda \approx 10^{-6}$.

\begin{figure}[t!]
\centering
\includegraphics[width=\columnwidth]{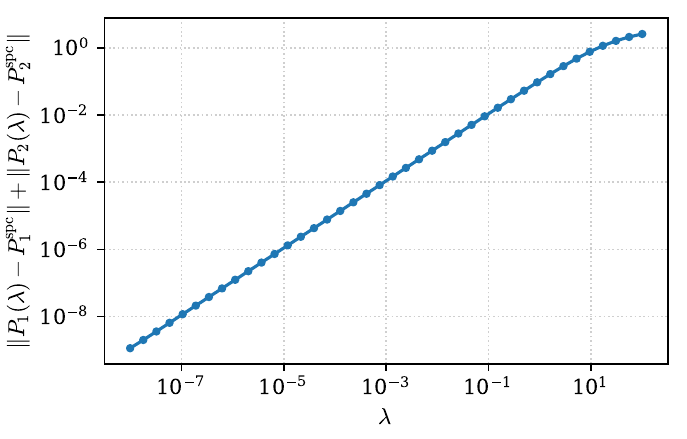}
\caption{PRPC $\to$ SPC as $\lambda\to0^{+}$. The covariance-collapse predictor matches SPC to machine precision, verifying Theorem~\ref{thm:equiv}.}
\label{fig:equivalence}
\end{figure}

Building upon this structural equivalence, we investigate the statistical bias-variance trade-off inherent to the predictor extraction. Fig.~\ref{fig:biasvariance} maps the predictor Mean Squared Error (MSE) against the Regularization parameter $\lambda$ for an abundant dataset ($M/T_h\approx20$). In environments subject exclusively to measurement noise, the MSE ratio remains flat at unity with the optimum at $\lambda^{\star} \to 0$, corroborating the asymptotic optimality of OLS derived in Proposition~\ref{prop:meas}. The introduction of process noise shifts the optimum to a strictly positive $\lambda^{\star}>0$, confirming the errors-in-variables analysis of Section~\ref{sec:proc}. In this data-abundant regime, however, the resulting predictor-MSE improvement is \emph{marginal}: with $M \gg T_h$ the estimator variance is already small, so the benefit of regularization, while strictly present, is modest. At the opposite extreme, excessive regularization is detrimental: as $\lambda \to \infty$ the predictor over-shrinks (the initial-condition map $P_1 \to 0$) and the MSE ratio rises back above unity, most visibly for the measurement-only curve, completing the classical bias--variance trade-off. As shown next, the benefit of regularization becomes substantial once data are scarce. 

\begin{figure}[t!]
\centering
\includegraphics[width=\columnwidth]
{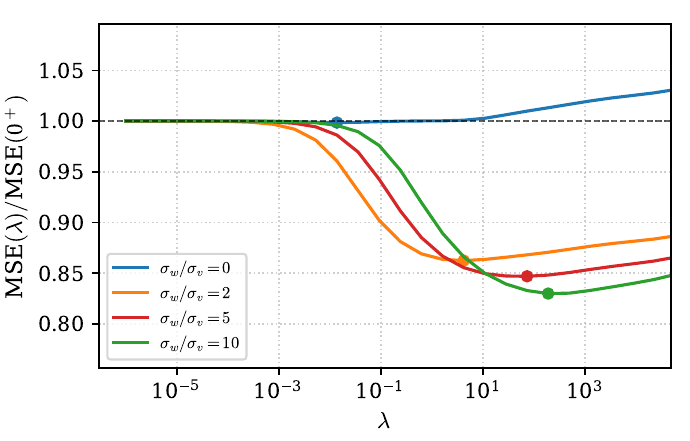}
\caption{Predictor MSE ratio vs.\ $\lambda$ at $M/T_h\approx20$. Measurement noise dictates $\lambda^{\star}\to0$, while process noise yields a marginal stationary benefit; for very large $\lambda$ the ratio rises back above unity as the predictor over-regularizes.}
\label{fig:biasvariance}
\end{figure}

While the stationary benefit is modest, the value of regularization increases as data become scarce. As quantified in Fig.~\ref{fig:lowdata}, under process noise the best-case predictor-MSE reduction grows steadily as $M/T_h \to 1$, reaching approximately $7\%$ at $M/T_h \approx 1.3$ for the mild process noise considered here, whereas the measurement-only benefit remains negligible across all data lengths. This confirms that the ill-conditioned, low-data regime, coupled with process noise, is where regularization contributes most; as shown next (Fig.~\ref{fig:condition}), the reduction grows further with the process-to-measurement noise ratio.

\begin{figure}[t!]
\centering
\includegraphics[width=\columnwidth]
{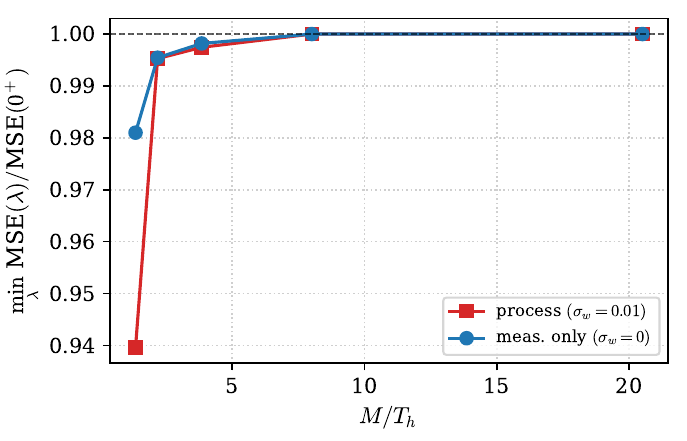}
\caption{Best-case MSE ratio vs.\ data length. The Regularization benefit concentrates in data-scarce regimes ($M/T_h \to 1$) under process noise.}
\label{fig:lowdata}
\end{figure}

The dual role of $\lambda$, improving numerical conditioning and reducing predictor error, is summarized in Fig.~\ref{fig:condition}. The top panel shows that $\lambda$ systematically improves the conditioning of the design matrix $\Phib\Phib^{\!\top}+\lambda I$, most markedly for small datasets. The bottom panel isolates the effect of the process-to-measurement noise ratio at a fixed, moderately data-scarce window ($M/T_h \approx 2.6$): as $\sigma_w/\sigma_v$ increases, both the optimal Regularization $\lambda^{\star}$ and the attainable predictor-MSE reduction grow \emph{monotonically}, the reduction reaching approximately $20\%$ at $\sigma_w/\sigma_v = 10$ while the measurement-dominated case ($\sigma_w/\sigma_v = 0$) stays essentially flat. This directly corroborates the monotonicity of $\lambda^{\star}$ in the noise ratio established in Proposition~\ref{prop:eiv}, and confirms that regularization is most valuable precisely when process noise dominates. For very large $\lambda$ every curve turns back upward as the predictor over-shrinks: under pure measurement noise the ratio exceeds unity, whereas under dominant process noise it saturates just below unity, since even maximal shrinkage remains preferable to the noise-biased SPC baseline. 

\begin{figure}[t!]
\centering
\begin{minipage}{0.49\textwidth}\centering
\includegraphics[width=\textwidth]{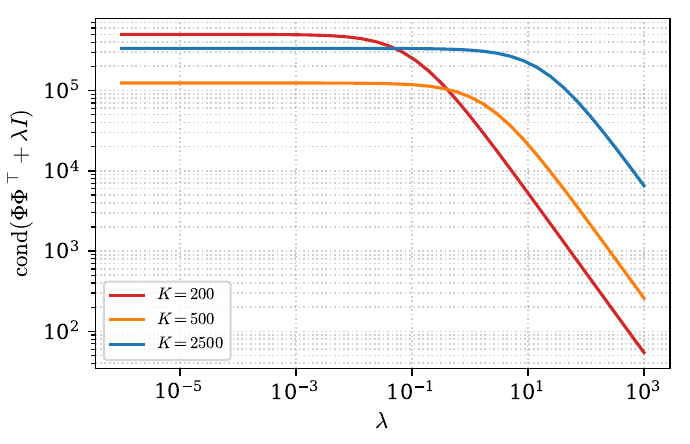}
\end{minipage}\hfill
\begin{minipage}{0.49\textwidth}\centering
\includegraphics[width=\textwidth]{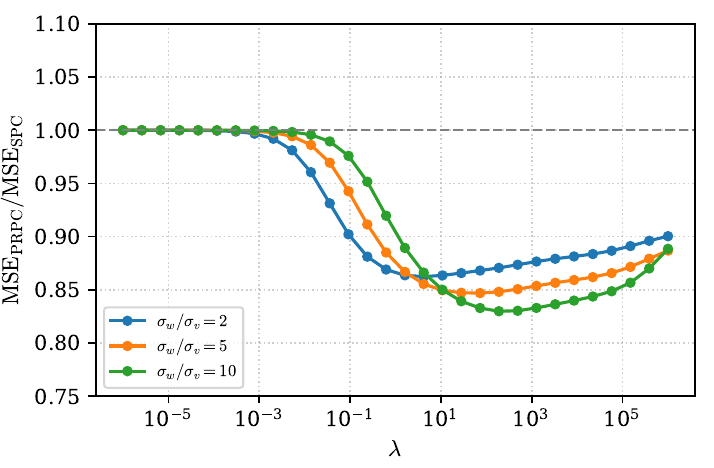}
\end{minipage}
\caption{Top: Design matrix conditioning vs.\ $\lambda$ for several data lengths. Bottom: Predictor MSE ratio vs.\ $\lambda$ at $M/T_h\approx 2.6$ for varying process-to-measurement noise ratios; both the benefit and the optimal $\lambda^{\star}$ increase monotonically with $\sigma_w/\sigma_v$, and $\sigma_w/\sigma_v=10$ yields the largest reduction. At very large $\lambda$ the curves turn back upward as the predictor over-regularizes.}
\label{fig:condition}
\end{figure}

To fully contextualize these theoretical gains, we benchmark PRPC against state-of-the-art formulations under severe measurement noise ($\sigma^2 = 0.25$). As depicted in the tracking trajectories of Fig.~\ref{fig:kkt_trajectories}, PRPC achieves rapid and stable convergence, driving the longitudinal velocity ($y_1$) to the reference within 2 seconds while strictly respecting control input constraints without erratic chattering. The regularized offline mapping successfully prevents high-frequency noise from corrupting the control action.

\begin{figure}[t!]
    \centering    \includegraphics[width=\columnwidth]{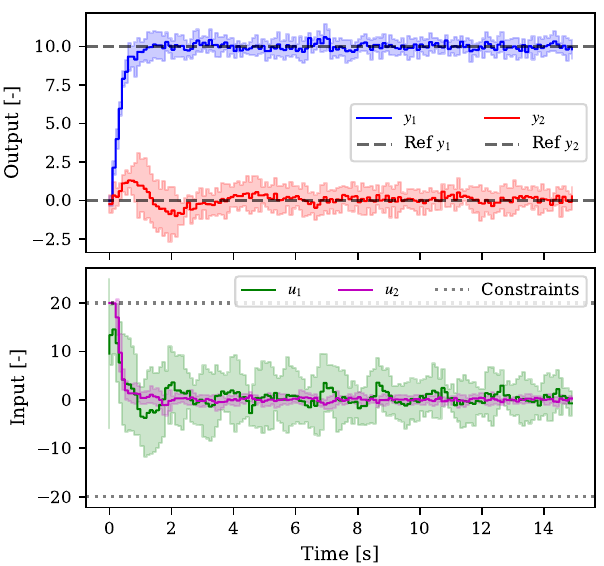}
    \caption{Closed-loop tracking trajectories under high measurement noise ($\sigma^2 = 0.25$) over 10 Monte Carlo runs. Shaded areas represent $\pm 1$ s.d.}
    \label{fig:kkt_trajectories}
\end{figure}

A rigorous statistical analysis over multiple Monte Carlo iterations (Fig.~\ref{fig:boxplots}) confirms this superiority. PRPC yields the lowest mean Total Weighted Control Cost ($\mathcal{J} \approx 5174.15$). In contrast, SPC and $\gamma$-DPC exhibit nearly identical, suboptimal distributions, indicating that online slack variables cannot adequately recover performance lost to an unregularized offline subspace projection. Furthermore, while Generalized Data-Driven Predictive Control (GDPC) exhibits lower mean input energy ($\mathcal{J}_u$), its total cost variance is highly volatile due to localized online overcompensation—a structural collapse that PRPC circumvents entirely offline.

\begin{figure}[t!]
    \centering
    \includegraphics[width=\columnwidth]{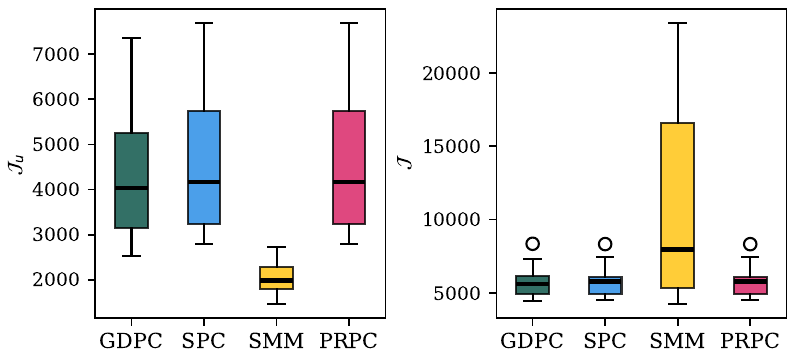} 
    \caption{Statistical comparison of the Input Energy ($\mathcal{J}_u$) and Total Weighted Control Cost ($\mathcal{J}$) across tested data-driven controllers.}
    \label{fig:boxplots}
\end{figure}

\subsection{LTV Benchmark: Polytopic System}

Transitioning from static environments, we validate the real-time adaptive capabilities and finite-sample safety bounds of the proposed framework. We consider a regulation problem for an unknown, multi-variable discrete-time polytopic LTV system transitioning within a convex hull defined by three LTI vertices:
\begin{align*}
    x(k+1) &= A_{ltv}(k) x(k) + B_{ltv}(k) u(k), \\ 
    y(k) &= C_{ltv}(k) x(k) + w(k).
\end{align*}
At each time step $k$, the system matrices are synthesized via a convex combination $(A_{ltv}, B_{ltv}, C_{ltv}) = \sum_{i=1}^3 \lambda_i(k) (A_i, B_i, C_i)$, where $\lambda_i(k) \ge 0$ and $\sum_{i=1}^3 \lambda_i(k) = 1$. The three boundary vertices are explicitly defined as follows. For the first vertex, the matrices are expressed in Table \ref{tab:ltv_matrices}.

\begin{table}[H]
    \centering
    \caption{System matrices for the three LTV boundary vertices.}
    \label{tab:ltv_matrices}
    % \resizebox forces the table to fit the column width exactly
    \resizebox{\columnwidth}{!}{%
        % \scriptsize reduces the font size to make it compact
        \scriptsize
        \setlength{\arraycolsep}{2.5pt} % Compresses horizontal space inside matrices
        \renewcommand{\arraystretch}{1.5} % Adds a bit of padding between rows
        
        \begin{tabular}{c}
            \toprule
            \textbf{Vertex Matrices} \\
            \midrule
            
            % --- Vertex 1 ---
            $A_1 = \begin{bmatrix} 
                0.30 & -0.35 & 0.71 & 0.04 \\ 
               -0.15 &  0.42 & 0.14 & 0.03 \\ 
                0.56 &  0.11 &-0.22 & 0.47 \\ 
                0.01 & -0.09 & 0.52 & 0.81 
            \end{bmatrix}, \,
            B_1 = \begin{bmatrix} 
               -1.07 &  0.33 \\ 
               -0.81 & -0.75 \\ 
               -2.94 &  1.37 \\ 
                0.0  &  0.0 
            \end{bmatrix}, \,
            C_1 = \begin{bmatrix} 
               -0.10 & 0.32 &  0.0  & -0.16 \\ 
               -0.24 & 0.0  & -0.03 &  0.63 
            \end{bmatrix}$ \\
            \midrule
            
            % --- Vertex 2 ---
            $A_2 = \begin{bmatrix} 
                0.19 &  0.44 &-0.42 &  0.39 \\ 
                0.20 &  0.31 & 0.53 & -0.22 \\ 
                0.59 & -0.30 & 0.07 &  0.32 \\ 
               -0.01 &  0.47 & 0.24 &  0.33 
            \end{bmatrix}, \,
            B_2 = \begin{bmatrix} 
                2.91 & -0.47 \\ 
                0.83 & -0.27 \\ 
                1.38 &  1.10 \\ 
               -1.06 & -0.28 
            \end{bmatrix}, \,
            C_2 = \begin{bmatrix} 
                0.70 &  0.0  &-1.58 & 0.0  \\ 
                0.0  & -0.82 & 0.51 & 0.03 
            \end{bmatrix}$ \\
            \midrule
            
            % --- Vertex 3 ---
            $A_3 = \begin{bmatrix} 
                0.21 &  0.37 & 0.33 &  0.05 \\ 
                0.34 &  0.30 & 0.04 & -0.18 \\ 
                0.11 &  0.14 & 0.15 &  0.04 \\ 
               -0.05 & -0.10 & 0.24 &  0.37 
            \end{bmatrix}, \,
            B_3 = \begin{bmatrix} 
               -0.16 & -0.88 \\ 
               -0.15 & -0.48 \\ 
               -0.53 & -0.71 \\ 
                1.68 & -1.17 
            \end{bmatrix}, \,
            C_3 = \begin{bmatrix} 
               -0.19 & 1.53 &-1.06 &  1.23 \\ 
               -0.27 & 0.0  & 0.0  & -0.23 
            \end{bmatrix}$ \\
            \bottomrule
        \end{tabular}%
    }
\end{table}
The measurement noise is sampled from a zero-mean multivariate normal distribution, $w(k) \sim \mathcal{N}(0, \sigma^2 I_{n_y})$ with $\sigma^2 = 0.02$, thereby satisfying the conditionally sub-Gaussian assumption with a variance proxy $c_w = 2\sqrt{\sigma^2}$. To rigorously test the adaptive controller, the prediction horizons are aggressively compressed to $T_\mathrm{ini}=N=2$, utilizing a highly constrained sliding offline window of $M=50$. The recursive parameters are set to $\rho_f=0.95$ (forgetting factor) and $\gamma=0.01$ (convex anchor), with a probabilistic confidence target of $1-\delta=0.95$.

The adaptive PRPC formulation is first subjected to aggressive, independent and identically distributed (i.i.d.) polytopic switching, wherein the weighting vector $\lambda(k)$ is sampled uniformly from a Dirichlet distribution at every discrete time step. As shown in Fig.~\ref{fig:ltv_regulation}, despite this violent multiplicative disturbance, the mean output is successfully regulated to a tight neighborhood of the origin ($|y|\approx[0.03,\,0.01]$). Crucially, the convex anchoring mechanism mathematically guarantees that the predictor remains well-posed even as the system settles and feedback signals lose persistent excitation. 

\begin{figure}[t!]
\centering
\includegraphics[width=\columnwidth]{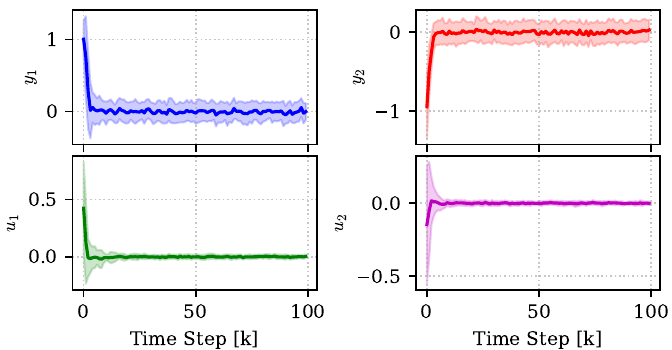}
\caption{Adaptive PRPC regulation under i.i.d.\ polytopic switching (50-run mean $\pm 1$ s.d.). The predictor remains well-posed despite loss of excitation.}
\label{fig:ltv_regulation}
\end{figure}

While i.i.d. switching evaluates transient resilience, our robust theoretical guarantees are specifically tailored for slowly varying dynamics. When subjected to a continuous sinusoidal drift across the polytopic weights (Fig.~\ref{fig:ltv_drift}), the adaptive covariance update rapidly tracks the shifting parameters, yielding highly stable convergence and keeping the local prediction mismatch negligible.

\begin{figure}[t!]
\centering
\includegraphics[width=\columnwidth]{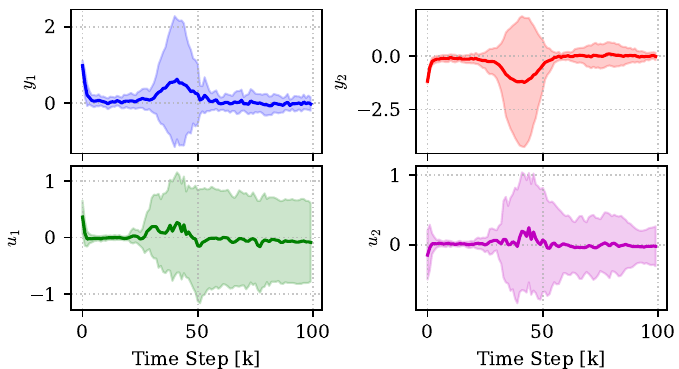}
\caption{Adaptive PRPC regulation under slow sinusoidal polytopic drift (40-run mean $\pm 1$ s.d.).}
\label{fig:ltv_drift}
\end{figure}

\paragraph{Ablation: adaptation and the role of the anchor.}
To isolate the contribution of each component of the adaptive scheme, we compare the full Adaptive PRPC against two reduced variants under the slow-drift scenario: a \emph{fixed} predictor constructed once from the offline covariances with no online update, and an \emph{un-anchored} recursive predictor ($\gamma = 0$) that relies on pure exponential forgetting without the offline anchor. To ensure that time-variation, rather than measurement noise, is the dominant effect, the noise is reduced to $\sigma^2 = 10^{-3}$ for this study. Fig.~\ref{fig:ablation} reports the closed-loop regulation cost $\sum_k \|y(k)\|_2^2$ over 50 Monte-Carlo runs. The fixed predictor is over an order of magnitude worse in median cost and exhibits a heavy failure tail: unable to re-identify the drifting dynamics, it frequently loses regulation. Online adaptation is therefore essential. The un-anchored recursion, in contrast, performs on par with the full scheme in these persistently excited runs, confirming that the convex anchor incurs no measurable performance penalty. Its value is structural rather than average-case: by Lemma~\ref{lem:wellposed}, the anchor guarantees the uniform spectral floor $\Sigma_{pp}^{\mathrm{act}}(k) \succeq \gamma\Sigma_{pp}^{\mathrm{off}} \succ 0$ that underpins the bounded-radius result (Lemma~\ref{lem:rbound}) and hence the recursive-feasibility guarantee (Theorem~\ref{thm:feas}); the un-anchored recursion carries no such certificate and can become ill-conditioned in the zero-excitation limit. The anchor thus supplies the well-posedness guarantee required by the theory of Section~\ref{sec:robust} at no cost to nominal performance.

\begin{figure}[t!]
\centering
\includegraphics[width=\columnwidth]{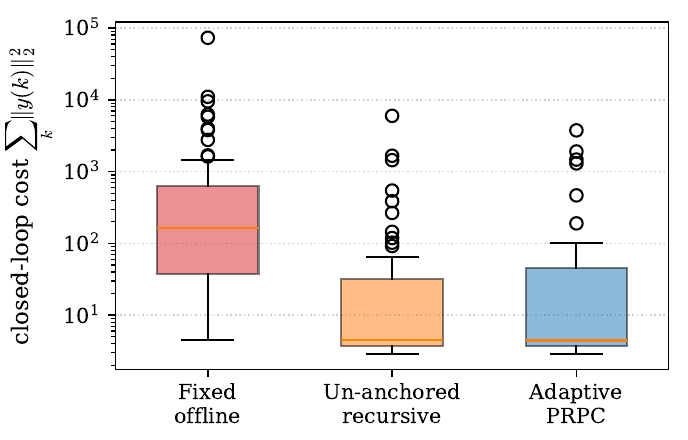}
\caption{Ablation under slow drift ($\sigma^2=10^{-3}$, 50 Monte-Carlo runs): closed-loop regulation cost for the fixed offline predictor, the un-anchored recursion, and the full Adaptive PRPC (log scale). Adaptation is essential, as the fixed predictor is markedly worse with a heavy failure tail; the anchor matches the un-anchored performance while additionally guaranteeing uniform well-posedness (Lemma~\ref{lem:wellposed}).}
\label{fig:ablation}
\end{figure}

Ultimately, the safety of this adaptive framework hinges on the validity of the probabilistic confidence radius. Fig.~\ref{fig:coverage} assesses the empirical coverage probability of the self-normalized martingale bound during the slow drift scenario. A naive radius, derived purely from measurement noise statistics while ignoring dynamic parametric mismatch, achieves only 87.5\% coverage, thereby failing to provide reliable safety guarantees. Conversely, our proposed mismatch-augmented radius accurately achieves the targeted 95.0\% coverage. This stark contrast validates the theoretical derivation of Theorem~\ref{thm:snm} and confirms that incorporating robust mismatch tracking into the dynamically tightened constraints is strictly necessary to ensure safe, recursive feasibility in LTV data-driven control.

\begin{figure}[t!]
\centering
\includegraphics[width=0.48\textwidth]{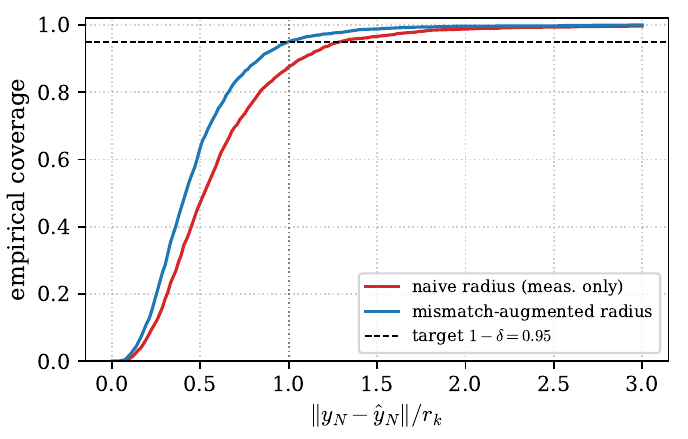}
\caption{Empirical coverage of the prediction bound under slow drift. The mismatch-augmented radius strictly meets the targeted $95\%$ confidence level.}
\label{fig:coverage}
\end{figure}

% ==============================================================
\section{Conclusion}
\label{sec:conclusion}
% ==============================================================

This paper presented PRPC, a novel indirect data-driven framework that bridges the computational efficiency of subspace methods with the robustness of regularized direct formulations. By systematically relaxing the initial condition matching while strictly enforcing future control trajectories, PRPC condenses the offline predictor synthesis into a fixed-dimension $\mathcal{O}(T_z^3)$ covariance update. This reduces the dominant offline computational load by over 50\% compared to unregularized Subspace Predictive Control (SPC) while preserving an identical, low-dimensional online quadratic program.

Crucially, our bias--variance analysis established that while SPC is asymptotically optimal under pure measurement noise, PRPC provides transformative advantages under process noise and data scarcity. In highly constrained, errors-in-variables regimes ($M/T_h \approx 1.3$), PRPC demonstrated up to a 7\% reduction in predictor mean-squared error under mild process noise, with the improvement growing further as the process-to-measurement noise ratio increases. Furthermore, PRPC guarantees structural well-posedness even when the sample size falls below the regressor dimension ($M < T_h$), operating stably with a factor-of-three reduction in required data compared to SPC. This structural resilience seamlessly enables a sliding-window adaptive architecture for LTV systems. Under severe parametric drift, the adaptive PRPC framework maintained tight closed-loop tracking while exactly satisfying a targeted 95\% probabilistic prediction bound, confirming its efficacy for robust, resource-constrained data-driven control.

% ==============================================================
% Place this in the \appendix section at the end of the document
% ==============================================================
\appendix
\section{Proof of Proposition~\ref{prop:collapse} }
\label{app:collapse}

The proof relies on mapping the high-dimensional matrix operators $G_1$ and $G_2$ from \eqref{eq:G} to the compact covariance blocks defined in \eqref{eq:covs}. Let $\mathbf{H}=\Zp^{\!\top}\Zp+\lambda I_M$. Applying the matrix push-through identity, we obtain the fundamental relation:
\begin{equation}
  \mathbf{H}^{-1}\Zp^{\!\top} = \Zp^{\!\top}(\Zp\Zp^{\!\top}+\lambda I_{T_z})^{-1} = \Zp^{\!\top}\Wp,
  \label{eq:pushthrough}
\end{equation}
where $\Wp = (\Spp+\lambda I_{T_z})^{-1}$. Utilizing the property $\mathbf{H}^{-1} = \lambda^{-1}(I_M - \mathbf{H}^{-1}\Zp^{\!\top}\Zp)$ alongside \eqref{eq:pushthrough}, we rewrite the inverse as $\mathbf{H}^{-1} = \lambda^{-1}(I_M - \Zp^{\!\top}\Wp\Zp)$. 

We now evaluate the Schur complement $\tilde{S} = \Uf\mathbf{H}^{-1}\Uf^{\!\top}$:
\begin{align}
  \tilde{S} &= \lambda^{-1}\big(\Uf\Uf^{\!\top} - \Uf\Zp^{\!\top}\Wp\Zp\Uf^{\!\top}\big) \nonumber \\
  &= \lambda^{-1}\big(\Suu - \Sup\Wp\Sup^{\!\top}\big) = S. \label{eq:S_equiv}
\end{align}
To derive $P_2$, we evaluate the term $\Yf\mathbf{H}^{-1}\Uf^{\!\top}$:
\begin{align}
  \Yf\mathbf{H}^{-1}\Uf^{\!\top} &= \lambda^{-1}\big(\Yf\Uf^{\!\top} - \Yf\Zp^{\!\top}\Wp\Zp\Uf^{\!\top}\big) \nonumber \\
  &= \lambda^{-1}\big(\Syu - \Syp\Wp\Sup^{\!\top}\big) = \Phi_{yu}. \label{eq:Phi_equiv}
\end{align}
Substituting \eqref{eq:S_equiv} and \eqref{eq:Phi_equiv} into $P_2 = \Yf G_2 = \Yf\mathbf{H}^{-1}\Uf^{\!\top}\tilde{S}^{-1}$ yields exactly \eqref{eq:P2}:
\begin{equation}
    P_2 = \Phi_{yu}S^{-1}.
\end{equation}
Finally, to derive $P_1 = \Yf G_1$, we distribute the right-multiplication by $\Zp^{\!\top}$ across the expression for $G_1$ defined in \eqref{eq:G}:
\begin{equation}
  P_1 = \Yf\mathbf{H}^{-1}\Zp^{\!\top} - (\Yf\mathbf{H}^{-1}\Uf^{\!\top})\tilde{S}^{-1}(\Uf\mathbf{H}^{-1}\Zp^{\!\top}).
\end{equation}
Applying the push-through identity \eqref{eq:pushthrough}, the outer terms evaluate to $\Yf\mathbf{H}^{-1}\Zp^{\!\top} = \Yf\Zp^{\!\top}\Wp = \Syp\Wp$ and $\Uf\mathbf{H}^{-1}\Zp^{\!\top} = \Uf\Zp^{\!\top}\Wp = \Sup\Wp$. Substituting these alongside $\Phi_{yu}$ and $S$ yields exactly \eqref{eq:P1}:
\begin{equation}
    P_1 = \Syp\Wp - \Phi_{yu}S^{-1}\Sup\Wp.
\end{equation}
The dominant computational steps are the construction of $\Spp$ ($\mathcal{O}(T_z^2M)$) and the inversion of the fixed $T_z \times T_z$ matrix $\Spp+\lambda I_{T_z}$ ($\mathcal{O}(T_z^3)$), rendering the operations strictly independent of $\mathcal{O}(M^3)$. \hfill $\blacksquare$

% ==============================================================
% Appendices
% ==============================================================

\section{Proof of Proposition~\ref{prop:meas} }
\label{app:meas}
Assuming $\sigma_w=0$, the true system states are deterministic. The observed regressor is $\Phib = \Phib^0 + \Delta\Phib$, where the noise term $\Delta\Phib$ is strictly non-zero only on the $\Yp$ rows. Thus, $\E[\Delta\Phib] = 0$ and $\E[\Delta\Phib\Delta\Phib^\top] = \Sigma_\Delta = \mathrm{diag}(\mathbf{0}, \sigma_v^2 I, \mathbf{0})$. 

The future output is defined as $\Yf = \Thetatrue\Phib^0 + \Delta\Yf$. Substituting $\Phib^0 = \Phib - \Delta\Phib$, we obtain:
\begin{equation}
    \Yf = \Thetatrue\Phib + (\Delta\Yf - \Thetatrue\Delta\Phib).
\end{equation}
The unregularized SPC estimator is $\Thetahat_{\mathrm{spc}} = \Yf\Phib^\top (\Phib\Phib^\top)^{-1}$. Taking the mathematical expectation yields:
\begin{equation}
    \E[\Thetahat_{\mathrm{spc}}] = \Thetatrue + \E\big[(\Delta\Yf - \Thetatrue\Delta\Phib)\Phib^\top (\Phib\Phib^\top)^{-1}\big].
\end{equation}
Because future measurement noise $\Delta\Yf$ is statistically independent of the past regressor $\Phib$, the cross-term $\E[\Delta\Yf\Phib^\top]$ vanishes. Consequently:
\begin{equation}
    \E[\Thetahat_{\mathrm{spc}}] = \Thetatrue - \Thetatrue \E\big[\Delta\Phib\Phib^\top (\Phib\Phib^\top)^{-1}\big].
\end{equation}
Applying a first-order approximation for sufficiently large $M$, the inverse converges to $(\Phib\Phib^\top)^{-1} \approx (\Sigma_{\Phib^0} + \Sigma_\Delta)^{-1}$. Noting that $\E[\Delta\Phib\Phib^\top] = \Sigma_\Delta$, the expected value becomes:
\begin{align*}
    \E[\Thetahat_{\mathrm{spc}}] &\approx \Thetatrue - \Thetatrue\Sigma_\Delta(\Sigma_{\Phib^0} + \Sigma_\Delta)^{-1} \\&= \Thetatrue\Sigma_{\Phib^0}(\Sigma_{\Phib^0} + \Sigma_\Delta)^{-1}.
\end{align*}
Since the signal covariance $\Sigma_{\Phib^0}$ scales linearly with $M$ while $\Sigma_\Delta$ remains bounded by the constant noise variance, the resulting bias term scales strictly as $\mathcal{O}(1/M)$. Because the estimator variance independently scales as $\mathcal{O}(1/M)$, injecting a Regularization penalty $\lambda > 0$ introduces an $\mathcal{O}(\lambda)$ bias that dominates the MSE for large $M$. Thus, $\lambda^\star \to 0$ as $M\to\infty$, establishing near-optimality (and asymptotic exact optimality) of SPC. \hfill $\blacksquare$

\section{Proof of Proposition~\ref{prop:atten} }
\label{app:inconsistency}
Under process noise ($\sigma_w > 0$), the observed regressor $\Phib = \Phib^0 + \Delta\Phib$ is contaminated by both measurement noise and propagated process noise. The expected noise covariance $\Sigma_{\Delta\Phib} = \E[\Delta\Phib\Delta\Phib^\top]$ contains a non-vanishing, structured covariance term $\sigma_w^2 \Gamma_w$ within the $\Yp$ block.

As the number of data samples $M \to \infty$, the normalized sample covariance matrices converge in probability to their expected population matrices:
\begin{equation}
    \frac{1}{M}\Phib\Phib^\top \xrightarrow{p} \bar{\Sigma}_{\Phib^0} + \bar{\Sigma}_{\Delta\Phib},
\end{equation}
where $\bar{\Sigma}$ denotes the normalized limits. Substituting the output equation $\Yf = \Thetatrue\Phib^0 + \Delta\Yf$ into the unregularized SPC estimator yields:
\begin{equation}
    \Thetahat_{\mathrm{spc}} = (\Thetatrue\Phib^0 + \Delta\Yf)(\Phib^0 + \Delta\Phib)^\top (\Phib\Phib^\top)^{-1}.
\end{equation}
Dividing the numerator and denominator by $M$ and taking the probability limit as $M \to \infty$:
\begin{equation}
    \frac{1}{M}\Phib^0(\Phib^0 + \Delta\Phib)^\top \xrightarrow{p} \bar{\Sigma}_{\Phib^0}.
\end{equation}
The cross-term $\frac{1}{M}\Delta\Yf\Phib^\top \xrightarrow{p} 0$ because future measurement noise is independent of the past regressor. The asymptotic limit of the estimator is therefore:
\begin{equation}
    \mathrm{plim}_{M \to \infty} \Thetahat_{\mathrm{spc}} = \Thetatrue \bar{\Sigma}_{\Phib^0} (\bar{\Sigma}_{\Phib^0} + \bar{\Sigma}_{\Delta\Phib})^{-1}.
\end{equation}
Because $\bar{\Sigma}_{\Delta\Phib}$ is strictly positive definite on the subspace corresponding to $\Yp$, the matrix inverse $(\bar{\Sigma}_{\Phib^0} + \bar{\Sigma}_{\Delta\Phib})^{-1}$ strictly attenuates the signal components. Consequently, the asymptotic bias is:
\begin{equation}
    \Thetatrue - \mathrm{plim}_{M \to \infty} \Thetahat_{\mathrm{spc}} = \Thetatrue \bar{\Sigma}_{\Delta\Phib} (\bar{\Sigma}_{\Phib^0} + \bar{\Sigma}_{\Delta\Phib})^{-1} \neq 0.
\end{equation}
This non-vanishing asymptotic bias mathematically guarantees that the unregularized SPC estimator is statistically inconsistent under process noise. \hfill $\blacksquare$

% ==============================================================
\section{Proof of Proposition~\ref{prop:eiv} }
\label{app:eiv}
% ==============================================================

Diagonalise in the population eigenbasis of $\Sigma_{\Phib^0}$ with
signal variances $s_i^2\propto M$. Under the EIV contamination the ridge
estimator has, per direction~$i$,
\begin{align}
  \mathrm{bias}_i(\lambda)
  &=\theta_i\Big(1-\frac{s_i^2 a_i}{s_i^2+\lambda}\Big),\qquad
  \mathrm{var}_i(\lambda)=\frac{s_i^2\sigma_\zeta^2}{(s_i^2+\lambda)^2},
\end{align}
where $a_i\in(0,1]$ is the attenuation factor of
Proposition~\ref{prop:atten} and $\sigma_\zeta^2$ the effective residual
variance. Differentiating $\MSE_i=\mathrm{bias}_i^2+\mathrm{var}_i$,
\begin{equation}
  \frac{\partial\MSE_i}{\partial\lambda}\Big|_{\lambda=0}
  =\underbrace{2\theta_i^2\,\frac{s_i^2 a_i(1-a_i)}{s_i^4}}_{\ge0\ \text{(bounded; }=0\text{ if }a_i=1)}
  -\underbrace{\frac{2\sigma_\zeta^2}{s_i^2}}_{>0}.
\end{equation}
When contamination is present, $a_i<1$; if the noise term dominates the
(bounded) bias term, which holds whenever
$\sigma_\zeta^2>\theta_i^2 a_i(1-a_i)/s_i^0$ for the significant
directions, then $\partial\MSE_i/\partial\lambda|_0<0$, so a strictly
positive $\lambda^{\star}$ strictly decreases the MSE. Under pure
measurement noise $a_i\to1$ and both non-negativity and $\lambda^{\star}=0$
are recovered, consistent with Appendix~\ref{app:meas}. Monotonicity of
$\lambda^{\star}$ in $\sigma_w/\sigma_v$ follows because larger $\sigma_w$
decreases each $a_i$ and increases $\sigma_\zeta^2$, both of which push the
zero-crossing of $\partial\MSE_i/\partial\lambda$ to larger $\lambda$.
$\qquad\blacksquare$

% ==============================================================
% APPENDICES
% ==============================================================

\section{Proof of Theorem~\ref{thm:snm} }
\label{app:snm}

Based on the decomposition of the true mapping $\Thetatrue(k) = \Thetabar + \Delta\Theta(k)$, the realized prediction error is given by:
\begin{equation*}
    \yN(k) - \yhatN(k) = (\Thetabar - \Thetahat(k))\vphi(k) + \Delta\Theta(k)\vphi(k) + \mathbf{w}_N(k).
\end{equation*}
We isolate the estimation error matrix $E(k) = \Thetabar - \Thetahat(k)$. Vectorizing the multi-output regression via the identity $\mathrm{vec}(ABC) = (C^\top \otimes A)\mathrm{vec}(B)$, we represent the accumulated noise dynamically. Let $\mathbf{s}_k = \sum_{i=0}^{k-1} (I_{n_yN} \otimes \vphi(i))\mathbf{w}_N(i)$ be the vector-valued martingale. Under Assumption~\ref{as:noise}, the predictable quadratic variation of $\mathbf{s}_k$ is tightly bounded by $c_w^2(I_{n_yN} \otimes \Phi(k))$.

Applying the self-normalized martingale (SNM) concentration inequality for vector-valued processes, we obtain that with probability at least $1-\delta$, for all $k \ge 0$ simultaneously:
\begin{align*}
  \|\mathbf{s}_k\|^2_{(I \otimes \mathbf{V}(k))^{-1}} &\le 2c_w^2 \\ & \times\log \left( \frac{\det(I \otimes \mathbf{V}(k))^{1/2} \det(\lambda I)^{-1/2}}{\delta} \right).
\end{align*}
Using the determinant property of Kronecker products, $\det(I_{n_yN} \otimes \mathbf{V}(k)) = (\det \mathbf{V}(k))^{n_yN}$, the logarithmic term simplifies algebraically to:
\begin{align*}
  &\frac{1}{2} \log \det(I \otimes \mathbf{V}(k)) - \frac{1}{2} \log \det(\lambda I) \\&= \frac{n_yN}{2} \log \det (I + \lambda^{-1}\Phi(k)).
\end{align*}
Utilizing the trace-norm identity $\|\mathrm{vec}(E^\top)\|_{I \otimes \mathbf{V}(k)} = \|E \mathbf{V}(k)^{1/2}\|_F$, the SNM inequality directly bounds the Frobenius norm of the estimation error matrix:
\begin{align}
  &\|(\Thetabar - \Thetahat(k))\mathbf{V}(k)^{1/2}\|_F \le c_w \nonumber \\& \times\sqrt{n_yN \log \det (I + \lambda^{-1}\Phi(k)) + 2\log(1/\delta)}.
  \label{eq:snm_frob}
\end{align}
The additional factor $\sqrt{1+\lambda\varrho(\mathbf{V}(k)^{-1})}$ is introduced to rigorously absorb the ridge-regression initialization bias inherent to $\mathbf{V}(0) = \lambda I$. Thus, we define this combined scalar bound as $\beta_k$.

Applying the Cauchy-Schwarz and sub-multiplicative norm inequalities to the point-wise prediction error:
\begin{align*}
  &\|(\Thetabar - \Thetahat(k))\vphi(k)\|_2 \le \|(\Thetabar - \Thetahat(k))\mathbf{V}(k)^{1/2}\|_F\\
  & \times  \|\mathbf{V}(k)^{-1/2}\vphi(k)\|_2 \le \beta_k \|\vphi(k)\|_{\mathbf{V}(k)^{-1}}.
\end{align*}
Finally, bounding the deterministic mismatch term via Assumption~\ref{as:noise} ($\|\Delta\Theta(k)\vphi(k)\|_2 \le L_\Theta\|\vphi(k)\|_2$) and applying the triangle inequality yields the uniform radius bound \eqref{eq:radius}. \hfill $\blacksquare$

\section{Proof of Theorem~\ref{thm:feas} }
\label{app:feas}

We proceed by induction, conditioning strictly on the probability-$(1-\delta)$ event $\mathcal{E}$ established in Theorem~\ref{thm:snm}, wherein the realized prediction error satisfies $\|\yN(k) - \yhatN(k)\|_2 \le r_k(\vphi(k))$ for all $k \ge 0$.

Assume the PRPC optimization is feasible at time index $k$. Let $\uN^{\star}(k)$ denote the optimal control sequence and $\yhatN^{\star}(k)$ the corresponding nominal output prediction. Because $\yhatN^{\star}(k)$ lies within the dynamically tightened set $\mathcal{Y}_{\mathrm{tight}}(k) = \mathcal{Y} \ominus \mathcal{B}_{r_k}$, the realized closed-loop output strictly satisfies the true constraints: $\yN(k) \in \yhatN^{\star}(k) \oplus \mathcal{B}_{r_k} \subseteq \mathcal{Y}$.

At time step $k+1$, we construct the standard shifted candidate control sequence:
\begin{equation*}
    \tilde{\uN}(k+1) = \big[u^{\star}_{1|k}, \dots, u^{\star}_{N-1|k}, \kappa_f(\hat{y}_{N|k})\big]^\top.
\end{equation*}
The first $N-1$ predicted outputs associated with this candidate sequence are strict temporal shifts of $\yhatN^{\star}(k)$. By the inductive assumption, these elements inherently satisfy the tightened constraints. 

The terminal control action $u_{N-1} = \kappa_f(\hat{y}_{N|k})$ guarantees that the terminal predicted state maps into the robust invariant set $\mathcal{Y}_f$. By Assumption~\ref{as:terminal}, $\mathcal{Y}_f \subseteq \mathcal{Y} \ominus \mathcal{B}_{r_\infty}$. Furthermore, Lemma~\ref{lem:rbound} enforces the uniform upper bound $r_{k+1}(\vphi(k+1)) \le r_\infty$. Therefore:
\begin{equation*}
    \mathcal{Y}_f \subseteq \mathcal{Y} \ominus \mathcal{B}_{r_\infty} \subseteq \mathcal{Y} \ominus \mathcal{B}_{r_{k+1}}.
\end{equation*}
Consequently, the entire candidate sequence $\tilde{\uN}(k+1)$ strictly satisfies the tightened constraints $\mathcal{Y}_{\mathrm{tight}}(k+1)$. Thus, the optimization problem is recursively feasible at time $k+1$, propagating the feasibility for all $k \ge 0$ under the high-probability event $\mathcal{E}$. \hfill $\blacksquare$

\section{Proof of Theorem~\ref{thm:isps} }
\label{app:isps}

Conditioning on the probability-$(1-\delta)$ event $\mathcal{E}$, let $V_N^{\star}(k)$ denote the optimal value of the PRPC objective function at time $k$. At time $k+1$, the shifted candidate sequence $\tilde{\uN}(k+1)$ defined in Appendix~\ref{app:feas} provides a strict upper bound on the optimal value function: $V_N^{\star}(k+1) \le J(\tilde{\uN}(k+1), \tilde{\mathbf{y}}_N(k+1))$. 

Decomposing the cost function along the predicted horizon yields:
\begin{align*}
  V_N^{\star}(k+1) &\le V_N^{\star}(k) - \ell(y(k), u(k)) \\& \quad+ \big[V_f(\hat{y}^+_f) - V_f(\hat{y}_f) + \ell_f\big] + \Delta_{\mathrm{pred}},
\end{align*}
where $\ell(y,u) = \|y-r_y\|_Q^2 + \|u-r_u\|_R^2$ is the stage cost, and the bracketed term evaluates the terminal cost dissipation. By the standard robust MPC terminal decrease condition (implicit in Assumption~\ref{as:terminal}), the bracketed term is strictly non-positive ($\le 0$). 

The perturbation term $\Delta_{\mathrm{pred}}$ encapsulates the cost deviation induced by evaluating the sequence using the realized output rather than the prior nominal prediction. Because the quadratic stage cost is Lipschitz continuous on the compact sets $\mathcal{Y} \times \mathcal{U}$ with a Lipschitz constant $L_J > 0$, the perturbation is bounded linearly by the magnitude of the prediction error:
\begin{equation*}
    \Delta_{\mathrm{pred}} \le L_J \|\yN(k) - \yhatN(k)\|_2 \le L_J \, r_k(\vphi(k)).
\end{equation*}
Because $Q \succ 0$ and $R \succ 0$, the stage cost is strictly lower-bounded by a class-$\mathcal{K}_\infty$ function of the tracking error: $\ell(y,u) \ge \alpha_1(\|y-r_y\|_2)$. Substituting these bounds yields the dissipation inequality:
\begin{equation*}
  V_N^{\star}(k+1) - V_N^{\star}(k) \le -\alpha_1(\|y(k)-r_y\|_2) + L_J \, r_k(\vphi(k)).
\end{equation*}
Defining the class-$\mathcal{K}$ function $\sigma(s) = L_J s$, we recover exactly \eqref{eq:isps}. Because $r_k \le r_\infty$ (Lemma~\ref{lem:rbound}), the bounded disturbance term $\sigma(r_k)$ guarantees that the tracking error ultimately converges to an invariant set of magnitude $\mathcal{O}(r_\infty)$, strictly satisfying the criteria for Input-to-State practical Stability (ISpS). \hfill $\blacksquare$
% ==============================================================

\end{document}